\documentclass[11pt]{article}
\usepackage{preamble}
\title{Quick Minimization of Tardy Processing Time \\ on a Single Machine}
\author{Baruch Schieber \thanks{Dept. of Computer Science, New Jersey Institute of Technology, Newark, NJ 07102, USA. \email{sbar@njit.edu}} \hspace{4em} Pranav Sitaraman \thanks{Edison Academy Magnet School, Edison, NJ 08837, USA. \email{sitaraman.pranav@gmail.com}}}
\date{}
\begin{document}
\maketitle
\begin{abstract}
We consider the problem of minimizing the total processing time of tardy jobs on a single machine. This is a classical scheduling problem, first considered by [Lawler and Moore 1969], that also generalizes the Subset Sum problem. Recently, it was shown that this problem can be solved efficiently by computing {\MMS}-convolutions. The running time of the resulting algorithm is equivalent, up to logarithmic factors, to the time it takes to compute a {\MMS}-convolution of two vectors of integers whose sum is $\bigo{P}$, where $P$ is the sum of the jobs' processing times. We further improve the running time of the minimum tardy processing time computation by introducing a job ``bundling" technique and achieve a $\bigot{P^{2-1/\alpha}}$ running time, where $\bigot{P^\alpha}$ is the running time of a {\MMS}-convolution of vectors of size $P$. This results in a $\bigot{P^{7/5}}$ time algorithm for tardy processing time minimization, an improvement over the previously known $\bigot{P^{5/3}}$ time algorithm.
\end{abstract}

\section{Introduction}
\label{sec:intro}
The input to the Minimum Tardy Processing Time (\MTPT) Problem consists of $n$ jobs each of which is associated with a due date and processing time $p_i \in \mathbb{N}$. Consider a (nonpreemptive) schedule of these jobs on a single machine that can execute only one job at a time. A job is \emph{tardy} if it terminates after its due date. The {\MTPT} Problem is to find a schedule of the jobs that minimizes the total processing time of the tardy jobs. In the standard scheduling notation the {\MTPT} problem is denoted $1||\sum p_jU_j$.

Consider an instance of {\MTPT} in which all the jobs have the same due date $d$. Let $P=\sum_{j=1}^n p_j>d$. The decision whether the total processing time of the tardy jobs is exactly $P-d$ (which is optimal in this case) is equivalent to finding whether there exists a subset of the jobs whose processing time sums to $d$. This is equivalent to the Subset Sum problem. It follows that {\MTPT} is NP-hard. {\MTPT} is weakly NP-hard and Lawler and Moore~\cite{LM1969} gave an $\bigo{P\cdot n}$ time algorithm for this problem.

Bringmann~\etal~\cite{alg2020} introduced a new convolution variant called a {\MMS}-convolution. They gave an algorithm for {\MTPT} that uses {\MMS}-convolutions, and proved that up to logarithmic factors, the running time of this algorithm is equivalent to the time it takes to compute a {\MMS}-convolution of integers that sum to $\bigo{P}$. They also gave an $\bigot{P^{7/4}}$ time algorithm\footnote{The notation $\bigot{\cdot}$ hides all logarithmic factors.} for computing a {\MMS}-convolution of integers that sum to $\bigo{P}$, which results in an $\bigot{P^{7/4}}$ time algorithm for the {\MTPT} problem. Klein \etal~\cite{conv2022} further improved the algorithm for computing a {\MMS}-convolution and achieved an $\bigot{P^{5/3}}$ running time, and thus an $\bigot{P^{5/3}}$ time algorithm for the {\MTPT} problem.

A natural approach to further improve the {\MTPT} algorithm is by improving the running time of a {\MMS} computation. However, obtaining an $\smallot{P^{3/2}}$ time algorithm for computing a {\MMS}-convolution seems difficult as this would imply an improvement to the best known (and decades old) algorithm for computing a {\MMConv}-convolution~\cite{K1989}. We were able to ``break'' the $\bigot{P^{3/2}}$ barrier by introducing a job ``bundling" technique. Applying this technique in conjunction with the best known algorithm for computing a {\MMS}-convolution yields an $\bigot{P^{7/5}}$ time algorithm for the {\MTPT} problem. This algorithm outperforms Lawler and Moore’s algorithm~\cite{LM1969} in instances where $n=\smallomegat{P^{2/5}}$. In general, applying our technique in conjunction with an $\bigot{P^\alpha}$ time algorithm for computing a {\MMS}-convolution yields an $\bigot{P^{2-{1}/{\alpha}}}$ time for the {\MTPT} problem.

The rest of the paper is organized as follows. In Section~\ref{sec:prelim} we introduce our notations and describe the prior work that we apply in our algorithm. Section~\ref{sec:alg} specifies our algorithm, and Section~\ref{sec:concl} has some concluding remarks.

\section{Preliminaries}
\label{sec:prelim}
Our notations follow the notations in~\cite{alg2020}. The input to the {\MTPT} problem is a set of $n$ jobs $\calJ = \lrc{J_1, J_2, \ldots, J_n}$. Each job $J_j\in \calJ$ has due date $e_j \in \mathbb{N}$ and processing time $p_j \in \mathbb{N}$. Let $D_\#$ denote the number of distinct due dates, and denote the monotone sequence of distinct due dates by $d_1< d_2< d_3<  \cdots< d_{D_\#}$, with $d_0 = 0$. Let $\calJ_k \subseteq \calJ$ be the set of jobs with due date $d_k$. Let $D = \sum_{i = 1}^{D_\#} d_i$ and $P = \sum_{i = 1}^{n} p_i$. For any $I=\lrc{i,\ldots,j}$, where $1\le i\le j\le D_\#$, let $\calJ_I = \bigcup_{i \in I} \calJ_i$ and $P_I = \sum_{J_i\in \calJ_I} p_i$.

Recall that the goal is to schedule the jobs in $\calJ$ so that the total processing time of tardy jobs is minimized. Since we only consider non-preemptive schedules, any schedule $S$ corresponds to a permutation $\sigma_S: \{1,\ldots,n\} \to \{1,\ldots,n\}$ of the job indices. The completion time of job $J_j \in \calJ$ in schedule $S$ is $C_j = \sum_{\sigma_S(i) \leq \sigma_S(j)} p_i$, and $j$ is tardy in $S$ if $C_j > e_j$. Therefore, we can consider that our algorithm seeks to minimize $\sum_{J_j\in \calJ, C_j > e_j} p_j$.

Next, we recall the definition of convolutions and describe the techniques developed in~\cite{alg2020} and used by our algorithm.
Given two vectors $A$ and $B$ of dimension $n+1$ and two binary operations $\circ$ and $\bullet$, the $(\circ,\bullet)$-convolution applied on $A$ and $B$ results in a $2n+1$ dimensional vector $C$, defined as:
\[
C[k] = \bigcirc_{i=\max\lrc{0, k-n}}^{\min\lrc{k,n}} A[i]\bullet B[k-i],\ \forall\ k\in \lrc{0,\ldots,2n}.
\]

A {\MMS}-convolution applied on $A$ and $B$ results in a $2n+1$ dimensional vector $C$, defined as:
\[
C[k] = \max_{i=\max\lrc{0,k-n}}^{\min\lrc{k,n}} \min\lrc{A[i],B[k-i]+k},\ \forall\ k\in \lrc{0,\ldots,2n}.
\]

Bringmann~\etal~\cite{alg2020} apply an equivalent form of {\MMS}-convolution defined as
\[
C[k] = \max_{i=\max\lrc{0,k-n}}^{\min\lrc{k,n}} \min\lrc{A[i],B[k-i]-i},\ \forall\ k\in \lrc{0,\ldots,2n}.
\]
Below, we use this equivalent form as well.

Let $X$ and $Y$ be two integral vectors. Define the \emph{sumset} $X \oplus Y = \{ x+y : x \in X, y \in Y\}$. It is not difficult to see that the sumset can be inferred from a {\PMConv}-convolution of $X_1$ and $X_2$ which can be computed in $\bigot{P}$ time for $X,Y \subseteq \{0,\ldots,P\}$ as in~\cite{AHU1974}.

The set of all subset sums of entries of $X$, denoted $\calS(X)$, is defined as $\calS(X)= \{\sum_{x \in Z} x : Z \sse X\}$. These subset sums can be calculated in $\bigot{\sum_{x \in X}x}$ time by successive computations of sumsets~\cite{KX2017}. We note that we always have $0 \in \calS(X)$. Define the $t-$\emph{prefix} and $t-$\emph{suffix} of $\calS(X)$ as
$\pref(\calS,t)=\lrc{x\in \calS(X) \land x\le t}$ and $\suff(\calS,t)=\lrc{x\in \calS(X) \land x> t}$.

We say that a subset of jobs $\calJ'\sse \calJ$ can be scheduled \emph{feasibly} starting at time $t$ if there exists a schedule of these jobs starting at time $t$ such that all jobs are executed by their due date. Note that it is enough to check whether all jobs in $\calJ'$ are executed by their due date in the \emph{earliest due date first} (EDD) schedule of these jobs starting at $t$.

For a consecutive subset of indices $I=\{i_0,i_0+1,\ldots,i_1\}$, with $1\le i_0\le i_1 \le D_\#$, define an integral vector $M(I)$ as follows. The entry $M(I)[x]$ equals $-\infty$ if none of the subsets of jobs in $\calJ_I$ with total processing time exactly $x$ can be scheduled feasibly. Otherwise, $M(I)[x]$ equals the latest time $t$ starting at which a subset of jobs in $\calJ_I$ with total processing time $x$ can be scheduled feasibly. Applying the algorithm for {\MMS}-convolutions given in~\cite{conv2022}, we get an $\bigot{{P_I}^{5/3}}$ time algorithm for computing $M(I)$, where $P_I = \sum_{J_i\in \calJ_I} p_i$.

In addition to the algorithm that uses {\MMS}-convolutions, Bringmann~\etal~\cite{alg2020} gave a second algorithm for the {\MTPT} problem. The running time of this algorithm is $\bigot{P\cdot D_\#}$. We use a version of this algorithm in our algorithm and for completeness we describe it in Algorithm~\ref{alg:sumset-scheduler}.

\begin{algorithm}[!ht]
  \caption{The $\bigot{P\cdot D_\#}$ time algorithm}
  \label{alg:sumset-scheduler}
  \begin{algorithmic}[1]
      \State Let $d_{1} < \ldots < d_{D_{\#}}$ denote the different due dates of jobs in $\calJ$.
      \For{$i=1,\ldots,D_\#$}
        \State Compute $X_i = \{p_j : J_j \in \calJ_i\}$
        \State Compute $\calS(X_i)$
      \EndFor
      \State Let $S_0=\emptyset$.
      \For{$i=1,\ldots,D_\#$}
        \State Compute $S_i=S_{i-1} \oplus \mathcal{S}(X_i)$.
        \State Remove any $x \in S_i$ with $x > d_i$.
      \EndFor
      \State Return $P-x$, where $x$ is the maximum value in $S_{D_{\#}}$.
  \end{algorithmic}
\end{algorithm}

\section{The Algorithm}
\label{sec:alg}
We define job bundles by coloring due dates in \emph{red} and \emph{blue}. The blue due dates are the bundled ones.

Choose some $\delta\in (0,1)$. For each $k = 1, 2, \ldots D_\#$, color the due date $d_k$ red if $\sum_{J_i \in \calJ_k} p_i > P^{1 - \delta}$. To determine the bundles we repeat the following procedure until all due dates are colored.

Let $m$ be the largest index for which due date $d_m$ is not yet colored. Find the smallest $k < m$ that satisfies the following conditions.
\begin{itemize}[align=left]
  \item[\textbf{Condition 1:}]
  None of the due dates $d_k,\ldots, d_m$ are colored red.
  \item[\textbf{Condition 2:}]
  $\sum_{i = k}^{m} \sum_{J_j \in \calJ_i} p_j \leq P^{1 - \delta}$
\end{itemize}
Color all due dates $d_k, d_{k + 1}, \ldots, d_m$ blue and ``bundle" them into one group, denoted $\bundle{k}{m}$. We say that due date $d_k$ is the \emph{start} of the bundle and $d_m$ is the \emph{end} of the bundle.

\begin{lemma} \label{lem:bundle}
  The number of red due dates is  $\bigo{P^{\delta}}$ and the number of bundles is $\bigo{P^{\delta}}$.
\end{lemma}
\begin{proof}
  Clearly, there can be
  at most $P^{\delta}$ due dates with $\sum_{J_i \in \calJ_k} p_i > P^{1 - \delta}$. 
  Consider a bundle $\bundle{k}{m}$. Since $k < m$ is the smallest index that satisfies the two conditions above, it is either true that $d_{k - 1}$ is red or $\sum_{i = k - 1}^{m} \sum_{J_j \in \calJ_i} p_j > P^{1 - \delta}$.
  \begin{enumerate}
    \item[(i)] Since there are at most $P^{\delta}$ red due dates, there can be at most only $P^{\delta}$ bundles $\bundle{k}{m}$ for which $d_{k - 1}$ is red.
    \item[(ii)] Consider the sum $\sum_{i = k - 1}^{m} \sum_{J_j \in \calJ_i} p_j$, for a bundle $\bundle{k}{m}$. Note that $p_j$ of a job $J_j \in \calJ_{k - 1}\cup \calJ_{m}$ may appear in at most one more sum that corresponds to a different bundle, while $p_j$ of a job $J_j \in \bigcup_{i=k}^{m-1}\calJ_i$ cannot appear in any other such sum. Thus, the total of all sums cannot exceed $2P$. It follows that there are at most $2P^{\delta}$  bundles $\bundle{k}{m}$ for which $\sum_{i = k - 1}^{m} \sum_{J_j \in \calJ_i} p_j > P^{1 - \delta}$.
  \end{enumerate}
\end{proof}
Algorithm~\ref{alg:solve} called~\nameref{alg:solve}, given below, follows the structure of Algorithm~\ref{alg:sumset-scheduler} with additional processing of entire bundles that avoids processing each due date in the bundles individually. We prove later that processing a bundle takes $\bigot{P^{(1-\delta)\cdot\alpha}+P}$ time, where $\bigot{P^\alpha}$ is the running time of the algorithm needed for computing a {\MMS}-convolution. Processing each red due date takes $\bigot{P}$ time. Substituting $\delta= 1- \frac 1\alpha$ yields a total running time of $\bigot{P\cdot P^{1-1/\alpha}}=\bigot{P^{2-1/\alpha}}$.

\begin{algorithm}[!ht]
  \caption{\textsc{Solve($\calJ$)}} \label{alg:solve}
  \begin{algorithmic}[1]
    \State Let $T=\lrc{0}$
    \State For each red due date $d_i$, compute $X_i = \lrc{p_j : e_j = d_i}$ and $\mathcal{S}(X_i)$
    \For{$i = 1, \ldots, D_\#$}
      \If {$d_i$ is a red due date}
        \State Compute $T = T \oplus \mathcal{S}(X_i)$ \label{line:case1}
        \State Remove any $x \in T$ with $x > d_i$
      \ElsIf {$d_i$ is the end of some bundle $\bundle{k}{i}$}
        \State Let $I = \lrc{k,\ldots,i}$
        \State Compute the vector $M(I)$. \label{line:case2g}
        \State Let $S_i = \lrc{x \in \lrc{0,\ldots,P_I} : M(I)[x] \ne -\infty}$. \label{line:case2a}
        \If {$d_k-P_I \ge 0$}
          \State Let $T = T \cup (\pref(T,d_k-P_I)\oplus S_i)$. \label{line:case2b}
        \EndIf
        \State Let $M'$ be an integral vector of dimension $d_k$ and initialize $M' = -\infty$. \label{line:case2c}
        \State For each $x\in \suff(T,d_k-P_I)$, let $M'[x]= 0$
        \For {$y= 0,\ldots,d_k-1+P_I$} \label{line:case2e}
          \State Let $C[y] = \max_{x=0}^y \min\{M'[x],M(I)[y-x]-x\}$
        \EndFor \label{line:case2f}
        \State Let $T_i = \lrc{x \in \lrc{0,\ldots, d_k-1+P_I}: C[x] = 0}$ \label{line:case2h}
        \State Let $T = T \cup T_i$\label{line:case2d}
        \State Remove any $x \in T$ with $x > d_i$
      \EndIf
    \EndFor
    \State Return $P-x$, where $x$ is the maximum value in $T$
  \end{algorithmic}
\end{algorithm}

\begin{theorem} \label{thm:solve}
  Algorithm~\nameref{alg:solve} returns the longest feasible schedule that starts at $d_0$.
\end{theorem}
\begin{proof}
  Consider iteration $i$ of Algorithm~\nameref{alg:solve}, for an index $i$ such that either $d_i$ is a red due date or $d_i$ is the end of some bundle $\bundle{k}{i}$. To prove the theorem it suffices to prove that at the end of any such iteration $i$ the set $T$ consists of the processing times of all feasible schedules of jobs in $\bigcup_{j=1}^i \calJ_j$ that start at $d_0$. The proof is by induction. The basis is trivial since $T$ is initialized to $\{0\}$. Consider such an iteration $i$ and suppose that the claim holds for all iterations $i'<i$ such that either $d_{i'}$ is a red due date or $d_{i'}$ is the end of some bundle $\bundle{k'}{i'}$. We distinguish two cases.

  \noindent\textbf{Case 1:} $d_i$ is a red due date. In this case it must be that either $d_{i-1}$ is also a red due date or $d_{i-1}$ is the end of some bundle $\bundle{k'}{i-1}$. Thus, by our induction hypothesis, at the start of iteration $i$ the set $T$ consists of the processing times of all feasible schedules of subsets of jobs in $\bigcup_{j=1}^{i-1} \calJ_j$ that start at $d_0$. Since iteration $i$ sets $T = T\oplus S(X_i)$ (Line~\ref{line:case1}), the claim follows.

  \noindent\textbf{Case 2:} $d_i$ is the end of some bundle $\bundle{k}{i}$. By our induction hypothesis, at the start of iteration $i$ the set $T$ consists of the processing times of all feasible schedules of subsets of jobs in $\bigcup_{j=1}^{k-1} \calJ_j$ that start at $d_0$. Let $I=\lrc{k,\ldots,i}$. The maximum length of any feasible schedule of subsets of jobs in $\calJ_I$ is $P_I$. Since the earliest due date of these jobs is $d_k$ we are guaranteed that any such feasible schedule can start at any time up to (and including) $d_k-P_I$ (assuming that $d_k-P_I\ge 0$). By the definition of $M(I)$ the set $S_i = \lrc{x \in \lrc{0,\ldots,P_I} : M(I)[x] \ne -\infty}$ consists of the processing times of all feasible schedules of subsets of jobs in $\calJ_I$ (Line~\ref{line:case2a}). $\pref(T,d_k-P_I)$ consists of the processing times of all feasible schedules of subsets of jobs in $\bigcup_{j=1}^{k-1} \calJ_j$ that start at $d_0$ and end at any time up to (and including) $d_k-P_I$. Since iteration $i$ sets $T = T\cup (\pref(T,d_k-P_I) \oplus S_i)$ (Line~\ref{line:case2b}), after this line $T$ consists of all the feasible schedules of subsets of jobs in $\bigcup_{j=1}^{i} \calJ_j$ that start at $d_0$ and also satisfy the condition that the sum of the lengths of the jobs in $\bigcup_{j=1}^{k-1} \calJ_j$ that are scheduled is at most $d_k-P_I$.

  The set $T$ is still missing the lengths of all the feasible schedules of subsets of jobs in $\bigcup_{j=1}^{i} \calJ_j$ that start at $d_0$ in which the sum of the lengths of the jobs in $\bigcup_{j=1}^{k-1} \calJ_j$ exceeds $d_k-P_I$. These schedules are added to $T$ in Lines~\ref{line:case2c}--\ref{line:case2d} of \nameref{alg:solve}. Consider such a feasible schedule of length $y$ in which the length of the jobs in $\bigcup_{j=1}^{k-1} \calJ_j$ is some $x> d_k-P_I$, which implies that $M'[x]=0$. To complement the prefix of this schedule by a feasible schedule of a subset of jobs in $J_I$ that starts at $x$ and is of length $y-x$ we must have $M(I)[y-x] \ge x$ or $\min\{M'[x],M(I)[y-x]-x\}=0$. Lines~\ref{line:case2e}--\ref{line:case2f} of \nameref{alg:solve} check if such a feasible schedule exists.
\end{proof}

\begin{lemma} \label{lem:time}
  The running time of algorithm~\nameref{alg:solve} is $\bigot{P^{(1-\delta)\cdot\alpha}+P}\cdot P^\delta$.
\end{lemma}

\begin{proof}
  By Lemma~\ref{lem:bundle} the number of iterations that are not vacuous is $P^\delta$. It is not difficult to see that all operations other than the computation of the vectors $M(I)$, $C$, and $T_i$ take $\bigot{P}$ time. The vector $M(I)$ is computed as in~\cite{alg2020} in $\bigot{{P_I}^\alpha}$ time. The vector $C$ is also computed via a {\MMS}-convolution and thus its computation time is proportional to the sum of lengths of the vectors $M(I)$ and $M'$ (up to logarithmic factors). Naively, this sum of lengths is $d_k+P_I$. However, since $M'[x]=-\infty$ for all $x\le d_k-P_I$, we can ignore these entries and implement the convolution in $\bigot{{P_I}^\alpha}$ time. Since $M'[x]=-\infty$, for all $x\le d_k-P_I$, we have also $C[x]= -\infty$, and thus $T_i$ can be computed in $\bigo{P_I}$ time (Line~\ref{line:case2h}). Recall that by the definition of bundles $P_I \le P^{1-\delta}$. Thus, the lemma is proved.
\end{proof}

\section{Conclusions}
\label{sec:concl}
We have shown a $\bigot{P^{7/5}}$ time algorithm for tardy processing time minimization, an improvement over the previously known $\bigot{P^{5/3}}$ time algorithm. Improving this bound further is an interesting open problem. In general, by applying our job ``bundling" technique we can achieve a $\bigot{P^{2-1/\alpha}}$ running time, where $\bigot{P^\alpha}$ is the running time of a {\MMS}-convolution of vectors of size $P$. Since it is reasonable to assume that computing a {\MMS}-convolution requires $\bigomegat{P^{3/2}}$ time, our technique is unlikely to yield a $\smallot{P^{4/3}}$ running time. It will be interesting to see whether this running time barrier can be broken, and whether the {\MTPT} problem can be solved without computing a {\MMS}-convolution.

\bibliographystyle{plain}
\bibliography{biblo}

\begin{thebibliography}{1}

\bibitem{AHU1974}
Alfred~V. Aho, John~E. Hopcroft, and Jeffrey~D. Ullman.
\newblock {\em {The Design and Analysis of Computer Algorithms}}.
\newblock Addison-Wesley, 1974.

\bibitem{alg2020}
Karl Bringmann, Nick Fischer, Danny Hermelin, Dvir Shabtay, and Philip
  Wellnitz.
\newblock Faster minimization of tardy processing time on a single machine.
\newblock {\em arXiv}, 2020.

\bibitem{conv2022}
Kim-Manuel Klein, Adam Polak, and Lars Rohwedder.
\newblock On minimizing tardy processing time, max-min skewed convolution, and
  triangular structured {ILP}s.
\newblock {\em arXiv}, 2022.

\bibitem{KX2017}
Konstantinos Koiliaris and Chao Xu.
\newblock A faster pseudopolynomial time algorithm for subset sum.
\newblock In {\em Proceedings of the 28th Annual ACM-SIAM Symposium on Discrete
  Algorithms}, page 1062–1072, USA, 2017. Society for Industrial and Applied
  Mathematics.

\bibitem{K1989}
S.~R. Kosaraju.
\newblock Efficient tree pattern matching.
\newblock In {\em Proceedings of the 30th Annual Symposium on Foundations of
  Computer Science}, page 178–183, USA, 1989. IEEE Computer Society.

\bibitem{LM1969}
Eugene~L. Lawler and J.~M. Moore.
\newblock A functional equation and its application to resource allocation and
  sequencing problems.
\newblock {\em Management Science}, 16:77--84, 1969.

\end{thebibliography}
\end{document}